\documentclass[a4paper,UKenglish]{lipics-v2016} 
\usepackage{microtype}
\usepackage{lineno}
\title{Opinion Forming in Erdős–Rényi Random Graph and Expanders}
\titlerunning{Opinion Forming in Erdős–Rényi Random Graph and Expanders} 
\author[1]{Ahad N. Zehmakan}
\affil[1]{ETH Zurich, Switzerland\\
\texttt{abdolahad.noori@inf.ethz.ch}}
\authorrunning{Ahad N. Zehmakan} 
\EventEditors{Wen-Lian Hsu, Der-Tsai Lee, and Chung-Shou Liao}
\EventNoEds{3}
\EventLongTitle{The 29th International Symposium on Algorithms and Computation (ISAAC 2018)}
\EventShortTitle{ISAAC 2018}
\EventAcronym{ISAAC}
\EventYear{2018}
\EventDate{December 16--19, 2018}
\EventLocation{Jiaoxi, Yilan, Taiwan}
\EventLogo{}
\SeriesVolume{123}
\ArticleNo{168}

\begin{document}
\maketitle

\begin{abstract}
Assume for a graph $G=(V,E)$ and an initial configuration, where each node is blue or red, in each discrete-time round all nodes simultaneously update their color to the most frequent color in their neighborhood and a node keeps its color in case of a tie. We study the behavior of this basic process, which is called majority model, on the Erdős–Rényi random graph $\mathcal{G}_{n,p}$ and regular expanders. First we consider the behavior of the majority model on $\mathcal{G}_{n,p}$ with an initial random configuration, where each node is blue independently with probability $p_b$ and red otherwise. It is shown that in this setting the process goes through a phase transition at the connectivity threshold, namely $\frac{\log n}{n}$. Furthermore, we say a graph $G$ is $\lambda$-expander if the second-largest absolute eigenvalue of its adjacency matrix is $\lambda$. We prove that for a $\Delta$-regular $\lambda$-expander graph if $\lambda/\Delta$ is sufficiently small, then  the majority model by starting from $(\frac{1}{2}-\delta)n$ blue nodes (for an arbitrarily small constant $\delta>0$) results in fully red configuration in sub-logarithmically many rounds. Roughly speaking, this means the majority model is an ``efficient'' and ``fast'' density classifier on regular expanders. As a by-product of our results, we show regular Ramanujan graphs are asymptotically optimally immune, that is for an $n$-node $\Delta$-regular Ramanujan graph if the initial number of blue nodes is $s\leq \beta n$, the number of blue nodes in the next round is at most $\frac{cs}{\Delta}$ for some constants $c,\beta>0$. This settles an open problem by Peleg~\cite{peleg2014immunity}. 
 
\textit{Eligible for best student paper}.
\end{abstract}
\section{Introduction}
A social network, the graph of relationships among a group of individuals, plays a fundamental role as a medium for the spread of information, ideas, and influence among its members. For example, social media such as Facebook, Twitter, and Instagram have served as a crucial tool for communication and information disseminating in today's life. Recently, studying different social behaviors like how people form their opinion regarding a new product or an election or how the information spreads through a social network have attracted a substantial amount of attention. Many different processes, from bootstrap percolation~\cite{balogh2007bootstrap} to rumor spreading~\cite{chierichetti2010rumour}, have been introduced to model this sort of social phenomena. 

A considerable amount of attention has been devoted to the study of the majority-based models, like voter model, majority bootstrap percolation, and majority model. In the \emph{majority bootstrap percolation} for a given graph and an initial configuration where each node is blue or red, in each round all blue nodes update their color to the most frequent color in their neighborhood and red nodes stay unchanged. The main motivation behind the majority bootstrap percolation is to model monotone processes like rumor spreading, where a red/blue node corresponds to an informed/uninformed individual and an informed individual will always stay informed of the rumor. However, to analyze non-monotone processes like the diffusion of two competing technologies over a social network or opinion forming in a community, the\emph{ majority model} is considered where each node updates its color to the most frequent color in its neighborhood and keeps it unchanged in case of a tie. The blue/red color, for instance, could stand for positive/negative opinion regarding a reform proposal.

Let us first fix some notations and define the majority model formally. For a graph $G=(V,E)$ and a node $v\in V$, let $N(v):=\{u\in V: \{v,u\}\in E\}$ be the \emph{neighborhood} of $v$ and $d(v):=|N(v)|$ be the \emph{degree} of $v$. Furthermore, for a set $S\subseteq V$, we define $N_S(v):=N(v) \cap S$ and $N(S):=\bigcup_{v\in S}N(v)$. For two node sets $S,S'\subset V$, define $e(S,S'):=|\{(v,u):v\in S,u\in S', \{v,u\}\in E\}|$. We also write $n$ for the number of nodes in a graph $G=(V,E)$, i.e. $|V|$.  

A \emph{configuration} is a function $\mathcal{C}:V\rightarrow\{b,r\}$, where $b$ and $r$ represent blue and red. For a set $S\subseteq V$, $\mathcal{C}|_S=a$ means $\forall v\in S$, $\mathcal{C}(v)=a$ for color $a\in\{b,r\}$. For a given initial configuration $\mathcal{C}_0$, assume $\forall\ t\geq 1$ and $v\in V$, $\mathcal{C}_t(v)$ is equal to the color that occurs most frequently in $v$'s neighborhood in $\mathcal{C}_{t-1}$, and in case of a tie $v$ keeps its current color, i.e. $\mathcal{C}_{t}(v)=\mathcal{C}_{t-1}(v)$. This deterministic process is called the \emph{majority model}. For a given initial configuration $\mathcal{C}_0$, let $B(t)$ and $R(t)$ for $t\geq 0$ denote the set of blue and red nodes in $\mathcal{C}_t$.

Since for a graph $G$ there are $2^{n}$ possible configurations and the majority model is a deterministic process, by starting from any initial configuration, the process must eventually reach a cycle of configurations. The length of the cycle and the number of rounds the process needs to reach it are respectively called the \emph{period} and the \emph{consensus time} of the process. $2^{n}$ is a trivial upper bound on both the period and the consensus time of the process. However, Goles and Olivos~\cite{goles1981comportement} provided the tight upper bound of two on the period of the process, and Poljak and Turzik~\cite{poljak1986pre} showed the consensus time is upper-bounded by $\mathcal{O}(n^2)$, which is shown to be tight up to some poly-logarithmic factor by Frischknecht, Keller, and Wattenhofer~\cite{frischknecht2013convergence}. 

The majority model has been studied on different classes of graphs, like lattice~\cite{gartner2017color,schonmann1990finite,shao2009dynamic,gartner2017biased}, infinite lattice~\cite{fontes2002stretched}, random regular graphs~\cite{gartner2018majority}, and infinite trees~\cite{kanoria2011majority}, when the initial configuration is random, meaning each node is independently blue with probability $p_b$ and red otherwise (without loss of generality, we always assume $p_b\leq 1/2$). We are interested in the behavior of the process when the underlying graph is the \emph{Erdős–Rényi random graph} $\mathcal{G}_{n,p}$, where the node set is $[n]=\{1,\cdots,n\}$ and each edge is added with probability $p$ independently. It is worth to mention that several other dynamic processes also have been studied on $\mathcal{G}_{n,p}$, for instance rumor spreading by Fountoulakis, Huber, and Panagiotou~\cite{fountoulakis2010reliable}, bootstrap percolation by Coja-Oghlan, Feige, Krivelevich, and Reichman~\cite{coja2015contagious}, and interacting particle systems by Schoenebeck and Yu~\cite{schoenebeck2018consensus}.

We prove that in the majority model with $p_b\leq \frac{1}{2}-\omega(\frac{1}{\sqrt{np}})$ on $\mathcal{G}_{n,p}$ with $(1+\epsilon)p^*\le p$ for any constant $\epsilon>0$ and $p^*=\frac{\log n}{n}$, the process gets fully red in constant number of rounds asymptotically almost surely (for an $n$-node graph $G$ we say an event happens asymptotically almost surely (a.a.s.) if it happens with probability $1-o(1)$ as $n$ tends to infinity). We also argue the tightness of this result. This explains the experimental observations from~\cite{lima2008majority}.
  
Furthermore, it is shown that in the majority model on $\mathcal{G}_{n,p}$ with $p\leq (1-\epsilon)p^*$ (for any constant $\epsilon>0$) if $p_b= o(e^{np}/n)$, then the process gets fully red but it does not for $p_b=\omega(e^{np}/n)$ a.a.s.

Putting the two aforementioned results together implies that the process exhibits a threshold behavior at $p^*$. More precisely, for $p=(1+\epsilon)p^*$, if the initial density of blue nodes is slightly less than one half, namely $\frac{1}{2}-\omega(1/\sqrt{\log n})$, then the process gets fully red, but for $p=(1-\epsilon)p^*$, $p_b$ must be very close to zero, namely smaller than $e^{n(1-\epsilon)\frac{\log n}{n}}/n=\frac{1}{n^{\epsilon}}$, to guarantee that it gets fully red a.a.s. Even though the proofs of the above statements require some effort, the main intuition behind this phase transition simply comes from the fact that $p^*$ is the connectivity threshold for $\mathcal{G}_{n,p}$, that is $\mathcal{G}_{n,p}$ is connected and disconnected a.a.s. respectively for $(1+\epsilon)p^*\le p$ and $p\le (1-\epsilon)p^*$. 
  
For $(1+\epsilon)p^*\le p$ and $p_b\le \frac{1}{2}-\omega(\frac{1}{\sqrt{np}})$ we distinguish two cases of sparse, $p\le \frac{n^{\gamma}}{n}$, and dense, $\frac{n^{\gamma}}{n}<p$ for some small constant $\gamma>0$. We argue that in the sparse case a very close neighborhood of each node includes only a constant number of cycles a.a.s., meaning it has a tree-like structure. Building on this tree-like structure, we prove that after constantly many rounds the probability of being blue for each node is so small that the union bound over all nodes yields our desired result. For the dense case, we argue in the first round the number of blue nodes decreases to $\frac{n}{c'}$ a.a.s. for a large constant $c'$. Then, relying on the high edge density of the graph we show $s\leq \frac{n}{c'}$ blue nodes can create at most $s/n^{\frac{\gamma}{2}}$ blue nodes in the next round; thus the process gets fully red in constantly many rounds. 

For $p\le (1-\epsilon)p^*$ and $p_b=\omega(e^{np}/n)$, the idea is to show that there exist sufficiently many constant-size components so that initially there is a fully blue and a fully red component a.a.s., which guarantee the coexistence of both colors. For $p_b=o(e^{np}/n)$, we argue the blue density is small enough to show that in at most two rounds all nodes are red a.a.s.
  
So far we considered the random setting, but one might approach the model from an extremal point of view, which brings up the very well-studied concept of dynamic monopoly. For a graph $G=(V,E)$ and the majority model a set $D\subseteq V$ is a \emph{dynamic monopoly}, or shortly \emph{dynamo}, when the following holds: if in some configuration all nodes in $D$ are red (similarly blue) then the process eventually gets fully red (resp. blue), regardless of the colors of the other nodes. Though the concept of dynamo had been studied before, e.g. by Balogh and Pete~\cite{balogh1998random} and Schonmann~\cite{schonmann1992behavior}, it was introduced formally by Kempe, Kleinberg, and Tardos~\cite{kempe2003maximizing} and Peleg~\cite{peleg1998size} independently and motivated from two different contexts. The minimum size of a dynamo has been extensively studied on different graph classes, from the $d$-dimensional lattice, motivated from the literature of statistical physics, by Flocchini, Lodi, Luccio, Pagli, and Santoro~\cite{flocchini2004dynamic}, Balister, Bollobás, Johnson, and Walters~\cite{balister2010random}, and Jeger and Zehmakan~\cite{jeger2018dynamic} to planar graphs by Peleg~\cite{peleg2002local}. As a notable example, although it had been conjectured by Peleg~\cite{peleg2002local} that the minimum size of a dynamo in any $n$-node graph is in $\Omega(\sqrt{n})$, surprisingly Berger~\cite{berger2001dynamic} proved for any $n\in\mathbb{N}$ there is an $n$-node graph which has a constant-size dynamo, meaning a constant number of red nodes is sufficient to make the whole graph red. We study the minimum size of a dynamo in $\mathcal{G}_{n,p}$, and prove it is larger than $(\frac{1}{2}-\frac{c}{\sqrt{np}})n$ a.a.s. for some constant $c>0$.

As we discussed, in $\mathcal{G}_{n,p}$ and above the connectivity threshold if $p_b$ is slightly less than one half then the process reaches fully red configuration and the minimum size of a dynamo is close to $n/2$ a.a.s. This raises the notorious and well-studied problem of density classification. 
For a given graph $G$, in the \emph{density classification problem}~\cite{gacs1978one} the task is to find an updating rule so that for whatever initial configuration, the process gets fully red if the number of reds is more than blues initially, and fully blue otherwise. This is a very central problem in the literature of cellular automata and distributed computing since it is a good test case to measure the power of local computations in gathering global information. This problem turned to be hard, in the sense that Land and Belew~\cite{land1995no} proved such an updating rule does not exist even when the underlying graph is a cycle. Mustafa and Pekec~\cite{mustafa2001majority,mustafa2004listen} approached the problem from a different angle and asked for which classes of graphs the majority model, which is probably one of the most natural candidates, classifies the density, and they proved that it is the case for graphs which have at least $n/2$ nodes of degree $n-1$. These hardness results however did not stop the quest for the best, although imperfect, solutions and different weaker variants of the problem have been suggested. A natural way of relaxing the problem would be to require any configuration with less than $(\frac{1}{2}-\delta)n$ blue nodes for some small $\delta>0$ results in fully red configuration. What are the graphs for which the majority model classifies the density for reasonably small values of $\delta$? To address this question, we argue that regularity and expansion are two determining factors.

Expanders are graphs which are highly connected; meaning to disconnect a large part of the graph, one has to sever many edges. A standard algebraic way of characterizing the expansion of a graph $G$ is to consider  the second-largest absolute eigenvalue of its adjacency matrix, which is denoted by $\lambda(G)$. For a $\Delta$-regular graph $G$, $\lambda(G)\le\Delta$ and smaller $\lambda(G)$ implies better expansion. We show that in the majority model on a $\Delta$-regular graph $G$, any starting configuration satisfying $|B(0)|\le(\frac{1}{2}-\delta)n$, for some fixed but arbitrarily small $\delta>0$, results in fully red configuration in sub-logarithmically many rounds if $\lambda(G)/\Delta$ is sufficiently small. In other words, if initially all nodes have the same color (which could correspond to some information) and an adversary is allowed to corrupt the color of $(\frac{1}{2}-\delta)n$ number of nodes, there is a large class of graphs for which if the nodes simply apply the majority rule, they all retrieve the original color in sub-logarithmically many rounds. Roughly speaking, the majority model is an ``efficient'' and ``fast'' density classifier on regular expanders.

In a graph $G=(V,E)$ and the majority model for two sets $S,S'\subseteq V$, we say $S$ \emph{controls} $S'$ when the following holds: if $S$ is fully blue (similarly red) in some configuration $\mathcal{C}$, $S'$ will be fully blue (resp. red) in the next configuration. The main idea of our results is that in regular expander graphs the number of edges between any two node sets $S, S'$ is almost completely determined by their cardinality. This simple fact implies the number of nodes that a set can control is proportional to its size, meaning a small set of blue nodes cannot make a big part of the graph blue. Applying this argument iteratively and some careful computations lead into the above result on regular expanders. It seems expansion is not only a sufficient condition for such a behavior but also some sort of a necessary condition since otherwise there can exist a small node set $S$ so that each node in $S$ has at least half of its neighbors inside $S$. Thus, if $S$ is initially blue, it stays blue forever, regardless of other nodes.

Motivated from fault-local mending in distributed systems, where redundant copies of data are kept and the
majority rule is applied to overcome the damage caused by failures, Peleg~\cite{peleg2014immunity} defined the concept of immunity. An $n$-node graph $G$ is \emph{$(\alpha,\beta)$-immune} if any node set of size $s\le \beta n$ can control at most $\alpha s$ nodes in the majority model. Immunity and density classification are related in the sense that for an $(\alpha,\beta)$-immune graph with $0<\alpha,\ \beta<1$, $|B(0)|\le \beta n$ results in fully red configuration in $\mathcal{O}(\log_{1/\alpha} n)$ rounds. For a $\Delta$-regular graph and some constant $\beta>0$ the best achievable $\alpha$ is $\frac{c_2}{\Delta}$ for some constant $c_2>0$ because $s$ nodes can occupy the full neighborhood of at least $\lfloor\frac{s}{\Delta}\rfloor$ arbitrary nodes. A $\Delta$-regular graph is called \emph{asymptotically optimally immune} if  it is $(\frac{c_2}{\Delta},\beta)$-immune for some constants $c_2,\beta>0$. These graphs are interesting since they prevent a small number of malicious/failed processors to take over a big fraction of the underlying graph. Peleg proved for any $\Delta>c_1$ for some constant $c_1$ there exists an asymptotically optimally immune $\Delta$-regular graph (actually he left a logarithmic gap, which was closed by G\"artner and Zehmakan~\cite{gartner2018majority}, recently). These results are existential, but one might be interested in constructing asymptotically optimally immune $\Delta$-regular graphs. For $\Delta\ge \sqrt{n}$, Peleg established explicit
construction of such graphs by using symmetric block designs. He also asked ``It would be interesting to construct asymptotically optimally immune regular graphs of degrees smaller than $\sqrt{n}$ ''. We settle this problem exploiting a large family of Cayley graphs, called Ramanujan graphs. 

In Section~\ref{random graph}, we study the behavior of the majority model on the random graph $\mathcal{G}_{n,p}$, and then in Section~\ref{expanders} we present our results regarding regular expander graphs and density classification. The uninterested reader might directly jump into Section~\ref{expanders} since the sections are supposed to stand by their own.
\section{Erdős–Rényi Random Graph}
\label{random graph}
In this section, we first study the behavior of the majority model on $\mathcal{G}_{n,p}$ with an initial random configuration (where each node is independently blue with probability $p_b$ and red otherwise) above the connectivity threshold in Theorem~\ref{theorem 1} and below it in Theorem~\ref{theorem 4}. 
As a corollary of these results it is easy to see that the process goes through a phase transition: above the connectivity threshold if $p_b$ is slightly less than $1/2$, the process gets fully red but below it the value of $p_b$ must be very close to zero to guarantee that it gets fully red a.a.s. Then in Theorem~\ref{dynamo}, we prove the minimum size of a dynamo in $\mathcal{G}_{n,p}$ is larger than $(\frac{1}{2}-\frac{c}{\sqrt{np}})n$ a.a.s. for some constant $c>0$, that is $(\frac{1}{2}-\frac{c}{\sqrt{np}})n$ blue nodes cannot make the whole graph blue no matter how they are placed in the graph.  
 
Let us state two variants of the Chernoff bound which we will use several times later.
\begin{theorem}\cite{feller2008introduction}
\label{Chernoff}
Suppose $x_1,\cdots,x_n$ are independent Bernoulli random variables taking values in $\{0,1\}$ and let $X$ denote their sum, then

(i) $\mathbb{P}[(1+\epsilon')\mathbb{E}[X]\leq X]\leq e^{-\frac{\epsilon'^2\mathbb{E}[X]}{3}}$ and $\mathbb{P}[X\leq (1-\epsilon')\mathbb{E}[X]]\leq e^{-\frac{\epsilon'^2\mathbb{E}[X]}{2}}$ for $0\leq \epsilon'\leq 1$

(ii) $\mathbb{P}[(1+\epsilon')\mathbb{E}[X]\le X]\leq e^{-\frac{\epsilon'\mathbb{E}[X]}{3}}$ for $\epsilon'\ge 1$.
\end{theorem}
To prove Theorem~\ref{theorem 1}, we need Lemma~\ref{lemma 3}, which states in $\mathcal{G}_{n,p}$ the degree of each node is concentrated around its expectation. This can be proven by simply applying the Chernoff bound (for a formal proof see e.g.~\cite{janson2011random}).
\begin{lemma}
\label{lemma 3}
In $\mathcal{G}_{n,p}$ if $p\geq (1+\epsilon)\frac{\log n}{n}$ for some constant $\epsilon>0$, then for each node $v$ $\mathbb{P}[d(v)< \frac{np}{c^{\prime\prime}}]=o(\frac{1}{n})$ for some constant $c^{\prime\prime}>0$ (as a function of $\epsilon$).
\end{lemma}
The main idea behind the proof of Theorem~\ref{theorem 1} is to apply the fact that the edges of each node are distributed randomly all over the graph.
\begin{theorem}
\label{theorem 1}
In the majority model with $p_b\leq \frac{1}{2}-\omega(\frac{1}{\sqrt{np}})$ on $\mathcal{G}_{n,p}$ with $p\geq (1+\epsilon)\frac{\log n}{n}$ for $\epsilon>0$, the process gets fully red in constant number of rounds a.a.s.
\end{theorem}
\begin{proof}
We divide the proof into two parts of dense, which is $p\ge \frac{n^{\gamma}}{n}$, and sparse, which is $p<\frac{n^{\gamma}}{n}$ for a sufficiently small constant $\gamma>0$. 

\textbf{Dense case:}
We first show that in one round a.a.s. the number of blue nodes decreases to $n/c'$ for an arbitrarily large constant $c'$. Then, we prove $n/c'$ blue nodes disappear in constant number of rounds, no matter how they are placed in the graph.

We argue that for an arbitrary node $v$, $\mathbb{P}[\mathcal{C}_1(v)=b]=o(1)$, which implies the expected number of blue nodes in $\mathcal{C}_1$ is equal to $o(n)$. By applying Markov's inequality~\cite{feller2008introduction}, the number of blue nodes in $\mathcal{C}_1$ is less than $n/c'$ a.a.s. for an arbitrarily large constant $c'$. To compute the probability that node $v$ is blue in $\mathcal{C}_1$, consider an arbitrary labeling $u_1,\cdots,u_{d(v)}$ of $v$'s neighbors and define Bernoulli random variable $x_i$ for $1\leq i\leq d(v)$ to be 1 if and only if $\mathcal{C}_0(u_i)=r$. Assume random variable $d_r(v)$ denotes the number of red nodes in $v$'s neighborhood in $\mathcal{C}_0$; clearly, $\mathbb{E}[d_r(v)]=\sum_{i=1}^{d(v)}x_i=d(v)(1-p_b)$. Let $p_b=1/2-\delta$ for $\delta=\omega(1/\sqrt{np})$ then by applying the Chernoff bound (Theorem~\ref{Chernoff} (i)) we have
\begin{align*}
&\mathbb{P}[\mathcal{C}_1(v)=b]\ \leq\ \mathbb{P}[d_r(v)\leq d(v)/2]\leq\ \mathbb{P}[d_r(v)\leq (1-\delta)(\frac{1}{2}+\delta)d(v)]\ = \\
&\mathbb{P}[d_r(v)\leq (1-\delta)\mathbb{E}[d_r(v)]]
\ \leq\ e^{-\frac{\delta^2(1/2+\delta)d(v)}{2}}.
\end{align*}
Thus, for some positive constant $c^{\prime\prime}$, we have $\mathbb{P}[\mathcal{C}_1(v)=b|d(v)\geq \frac{np}{c^{\prime\prime}}]\leq e^{-\frac{\delta^2(1/2+\delta)np}{2c^{\prime\prime}}}=e^{-\omega(1)}$, where we used $\delta=\omega(1/\sqrt{np})$. Now, by applying Lemma~\ref{lemma 3},
\begin{align*}
&\mathbb{P}[\mathcal{C}_1(v)=b]=\mathbb{P}[\mathcal{C}_1(v)=b|d(v)\geq \frac{np}{c^{\prime\prime}}]\cdot \mathbb{P}[d(v)\geq  \frac{np}{c^{\prime\prime}}]+\\
&\mathbb{P}[\mathcal{C}_1(v)=b|d(v)< \frac{np}{c^{\prime\prime}}]\cdot \mathbb{P}[d(v)< \frac{np}{c^{\prime\prime}}]\leq e^{-\omega(1)}\cdot 1+1\cdot o(1)=o(1).
\end{align*}

Now, we prove any non-empty node set of size $s\leq n/c'$ controls at most $s/n^{\frac{\gamma}{2}}$ nodes a.a.s. This implies by starting from $n/c'$ blue nodes (regardless of how they are placed in the graph) the process gets fully red after at most $2/\gamma$ rounds (notice that $2/\gamma$ is a constant). Let $S$ be a set of size $s\leq n/c'$ and $S'$ be a set of size $s'=s/n^{\gamma/2}$. Since $\mathbb{E}[e(S',V\setminus S)]=s'(n-s)p$, by applying the Chernoff bound (Theorem~\ref{Chernoff} (i)) and using $p\ge \frac{n^{\gamma}}{n}$, $n-s\ge n/2$, and $s'=s/n^{\gamma/2}$ we have
\begin{equation}
\label{eq 1}
\mathbb{P}[e(S',V\setminus S)\le (1-\frac{1}{2})\mathbb{E}[e(S',V\setminus S)]]\leq e^{-\frac{\mathbb{E}[e(S',V\setminus S)]}{8}}= e^{-\frac{s'(n-s)p}{8}}\le e^{-\Theta(sn^{\frac{\gamma}{2}})}
\end{equation}

Similarly, since $\mathbb{E}[e(S',S)]=s'sp$ again by applying the Chernoff bound (Theorem~\ref{Chernoff} (ii))
\begin{equation}
\label{eq 2}
\mathbb{P}[e(S',S)\ge (1+(\frac{n}{4s}-1))\mathbb{E}[e(S',S)]]\leq e^{-(\frac{n}{4s}-1)\frac{\mathbb{E}[e(S',S)]}{3}}= e^{-(\frac{n}{4s}-1)\frac{s'sp}{3}}\le e^{-\Theta(sn^{\frac{\gamma}{2}})}
\end{equation}

Clearly, $\mathbb{P}[S\ \textrm{controls}\ S']\leq \mathbb{P}[e(S',V\setminus S)\le e(S',S)]$ because if $e(S',V\setminus S)> e(S',S)$ then there is at least one node in $S'$ which has more than half of its neighbors in $V\setminus S$. Furthermore, $(1+(\frac{n}{4s}-1))\mathbb{E}[e(S',S)]=\frac{n}{4s}s'sp=\frac{n}{4}s'p$ and by using $(n-s)\ge n/2$ we have $(1-\frac{1}{2})\mathbb{E}[e(S',V\setminus S)]=\frac{1}{2}s'(n-s)p\ge \frac{n}{4}s'p$. Thus by Equations~\ref{eq 1} and \ref{eq 2}, $\mathbb{P}[S\ \textrm{controls}\ S']\le 2e^{-\Theta(sn^{\gamma/2})}=e^{-\Theta(sn^{\gamma/2})}$ since $s\geq 1$.

By the union bound, the probability that there exits a set $S$ of size $s\le n/c'$ which controls a set of size $s/n^{\gamma/2}$ is bounded by
\[
\sum_{s=1}^{n/c'}{n \choose s}{n \choose s/n^{\gamma/2}} e^{-\Theta(sn^{\gamma/2})}\le \sum_{s=1}^{n/c'}n^{2s}e^{-\Theta(sn^{\gamma/2})}\leq \sum_{s=1}^{n/c'}(n^2e^{-\Theta(n^{\gamma/2})})^s.
\]
$(n^2e^{-\Theta(n^{\gamma/2})})^s$ is maximized for $s=1$ since $n^2e^{-\Theta(n^{\gamma/2})}<1$. Thus, the summation is upper-bounded by $\frac{n}{c'}n^2e^{-\Theta(n^{\gamma/2})}=o(1)$ which proves our claim.

\textbf{Sparse case:}
Let us first present the following proposition, which roughly speaking states that for small values of $p$, the close neighborhood of each node looks like a tree.

\textbf{Proposition 1}: In $\mathcal{G}_{n,p}$ with $p< \frac{n^{\gamma}}{n}$ for some small constant $\gamma>0$, a.a.s. there is no node which is in two different cycles of size 3 or 4.

To prove Proposition 1, it suffices to show a.a.s. there exits no subgraph with $4\leq k\leq 7$ nodes and $k+1$ edges. By the union bound, the probability of having such a subgraph is upper-bounded by $\sum_{k=4}^{7}{n \choose k}{k(k-1)/2 \choose k+1}p^{k+1}\leq \sum_{k=4}^{7}\Theta(n^k) \frac{n^{(k+1)\gamma}}{n^{k+1}}=o(1)$, where in the last step we used the fact that $\gamma$ is a sufficiently small constant (for instance $\gamma <1/8$). This finishes the proof of Proposition 1.

Now building on this tree-like structure and Lemma~\ref{lemma 3}, we prove the probability that an arbitrary node is blue after two rounds of the process is so small that the union bound over all nodes implies the process is fully red a.a.s.
Let $v$ be an arbitrary node and label its neighbors from $u_1$ to $u_{d(v)}$. We want to upper-bound $\mathbb{P}[\mathcal{C}_2(v)=b]$. For $1\le i\le d(v)$ let $u_{i}^{1},\cdots ,u_{i}^{d(u_i)-1}$ be the neighbors of $u_i$ except $v$. Define random variable $X_i$ to be the number of red nodes among $u_{i}^{1},\cdots, u_{i}^{d(u_i)-1}$ in $\mathcal{C}_0$. We say node $u_i$ is \emph{almost blue} in $\mathcal{C}_1$ if $X_i\le \frac{d(u_i)}{2}$ (notice if a node is blue in $\mathcal{C}_1$, it is also almost blue, but not necessarily the other way around).
Now, we bound $\mathbb{P}[X_i\le \frac{d(u_i)}{2}]$, which is the probability that $u_i$ is almost blue. Since $\mathbb{E}[X_i]=(d(u_i)-1)(1-p_b)$ for $p_b=\frac{1}{2}-\delta$ and $\delta=\omega(\frac{1}{\sqrt{np}})$, by applying the Chernoff bound (Theorem~\ref{Chernoff} (i)) we have
\[
\mathbb{P}[X_i\leq \frac{d(u_i)}{2}]\leq \mathbb{P}[X_i\le (1-\delta)(\frac{1}{2}+\delta)(d(u_i)-1)]=\mathbb{P}[X_i\le (1-\delta)\mathbb{E}[X_i]]\le e^{-\frac{\delta^2(1/2+\delta)(d(u_i)-1)}{2}}.
\]
Thus for any large constant $c^{\prime\prime}$, $\mathbb{P}[X_i\leq \frac{d(u_i)}{2}|d(u_i)\geq \frac{np}{c^{\prime\prime}}]\leq e^{-\frac{\delta^2(1/2+\delta)(\frac{np}{c^{\prime\prime}}-1)}{2}}=o(1)$ by $\delta=\omega(\frac{1}{\sqrt{np}})$. Now by applying Lemma~\ref{lemma 3}, we have
\begin{align*}
&p_i:=\mathbb{P}[X_i\le \frac{d(u_i)}{2}]=\mathbb{P}[X_i\le \frac{d(u_i)}{2}|d(u_i)\ge \frac{np}{c^{\prime\prime}}]\cdot\mathbb{P}[d(u_i)\geq \frac{np}{c^{\prime\prime}}]+\\
&\mathbb{P}[X_i\le \frac{d(u_i)}{2}|d(u_i)<\frac{np}{c^{\prime\prime}}]\cdot\mathbb{P}[d(u_i)<\frac{np}{c^{\prime\prime}}]\le o(1)\cdot 1+1\cdot o(1)\le \delta'
\end{align*}
for an arbitrarily small constant $\delta'>0$.

Now, we bound the probability $\mathbb{P}[\mathcal{C}_2(v)=b]$. Based on Proposition 1, a.a.s. every node, including $v$, is in at most one cycle of length three, say with $u_1$ and $u_2$, and in at most one cycle of length four, say with $u_3$ and $u_4$, and other $u_i$s share no neighbor except $v$ (see Figure~\ref{fig 1}). Let $Y$ denote the number of nodes among $u_{5},\cdots, u_{d(v)}$ which are almost blue in $\mathcal{C}_1$. Then, $\mathbb{P}[\mathcal{C}_2(v)=b]\le \mathbb{P}[Y\geq \frac{d(v)}{2}-4]$ because for $u_i$ to be blue in $\mathcal{C}_1$ it must be almost blue in $\mathcal{C}_1$ by definition and for $v$ to be blue in $\mathcal{C}_2$ it needs at least $\frac{d(v)}{2}-4$ blue nodes among $u_{5},\cdots, u_{d(v)}$. Notice that being almost blue and being blue are pretty much the same except being almost blue is not a function of the color of node $v$ (we some sort of assume node $v$ is blue in $\mathcal{C}_0$ and still the impact of this assumption is small enough to let us get our desired tail bound). This gives us the independence among $p_i$s for $4\leq i\leq d(u_i)-1$ (which we apply in the next step) because the only neighbor they share is $v$. Since we upper-bounded $p_i$ by $\delta'$,
\[
\mathbb{P}[\mathcal{C}_2(v)=b]\leq \mathbb{P}[Y\geq \frac{d(v)}{2}-4]=\sum_{j=\frac{d(v)}{2}-4}^{d(v)-4}{d(v)-4 \choose j}\delta'^{j}(1-\delta')^{d(v)-4-j}\le 2^{d(v)} \delta'^{\frac{d(v)}{2}-4}
\]
which is equal to $\frac{(2\sqrt{\delta'})^{d(v)}}{\delta'^4}$. Thus, $\mathbb{P}[\mathcal{C}_2(v)=b|d(v)\geq \frac{np}{c^{\prime\prime}}]\leq \frac{(2\sqrt{\delta'})^{\frac{np}{c^{\prime\prime}}}}{\delta'^4}$ which is less than $e^{-2np}$ by selecting $\delta'$ sufficiently small. Furthermore, $e^{-2np}=o(\frac{1}{n})$  by $p\ge (1+\epsilon)\frac{\log n}{n}$.  Now by Lemma~\ref{lemma 3},
\begin{align*}
&\mathbb{P}[\mathcal{C}_2(v)=b]=\mathbb{P}[\mathcal{C}_2(v)=b|d(v)\ge \frac{np}{c^{\prime\prime}}]\cdot\mathbb{P}[d(v)\geq \frac{np}{c^{\prime\prime}}]+\\
&\mathbb{P}[\mathcal{C}_2(v)=b|d(v)<\frac{np}{c^{\prime\prime}}]\cdot\mathbb{P}[d(v)<\frac{np}{c^{\prime\prime}}]\le o(\frac{1}{n})\cdot 1+1\cdot o(\frac{1}{n})=o(\frac{1}{n}).
\end{align*}
The union bound implies a.a.s. there is no blue node in $\mathcal{C}_2$. 
\begin{figure}[t]
\centering
\includegraphics{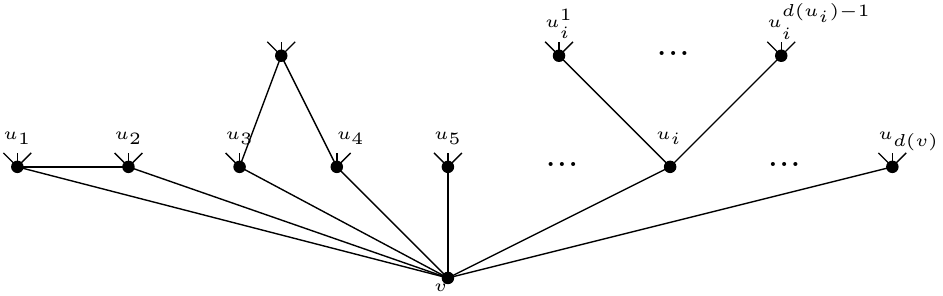}
\caption{\centering The neighborhood of node $v$.\label{fig 1}}
\end{figure} 
\end{proof}
Regarding the tightness of the result of Theorem~\ref{theorem 1}, notice that it does not hold if we replace $\omega(\frac{1}{\sqrt{np}})$ with $\frac{c}{\sqrt{np}}$ for any constant $c$. For $p=1$, which corresponds to the complete graph, if we color each node blue independently with probability $p_b=\frac{1}{2}-\frac{c}{\sqrt{n}}$ and red otherwise for some constant $c>0$, then by Central Limit Theorem~\cite{feller2008introduction} the probability that more than half of the nodes are blue is a positive constant. This implies the process gets fully blue after one round with some positive constant probability.
\begin{theorem}
\label{theorem 4}
In the majority model on $\mathcal{G}_{n,p}$ with $p\leq (1-\epsilon)\frac{\log n}{n}$ for $\epsilon>0$, a.a.s.

(i) $p_b=\omega(e^{np}/n)$ results in the coexistence of both colors

(ii) $p_b=o(e^{np}/n)$ results in fully red configuration.
\end{theorem}
\begin{proof}
Proof of part (i): Notice that a blue/red isolated node never changes its color in majority model. Thus, it suffices to show for $P_b=\omega(e^{np}/n)$, a.a.s. there is a blue and a red isolated node in the initial configuration. We discuss the blue case and the proof carries on analogously for red. 

Let random variable $X$ denote the number of blue isolated nodes in $\mathcal{C}_0$. Consider an arbitrary labeling $v_1,\cdots,v_n$ on the nodes and define the Bernoulli random variable $x_i$, for $1\leq i\leq n$, to be one if and only if node $v_i$ is isolated and blue in $\mathcal{C}_0$. Clearly, $X=\sum_{i=1}^{n}x_i$ and $\mathbb{P}[x_i=1]=p_b(1-p)^{n-1}$. Thus, by linearity of expectation $\mathbb{E}[X]=np_b(1-p)^{n-1}$. By applying the estimate $1-x\geq e^{-x-x^2}$ for $0\leq x\leq 1/2$, plugging in $p_b=\omega(e^np/n)$, and using the fact that $e^{np^2}\le e^{\frac{\log^2n}{n}}\leq e$, we have $\mathbb{E}[X]\geq n\ \omega (\frac{e^{np}}{n})\ e^{-np-np^2}=\omega(1)$.
Now, we argue that $\mathrm{Var}(X)=o(\mathbb{E}[X]^2)$, which then simply by applying Chebychev's inequality~\cite{feller2008introduction} implies $\mathbb{P}[X=0]\leq\mathrm{Var}(X)/\mathbb{E}[X]^2=o(1)$. Therefore, a.a.s. there exist a blue and a red isolated node in $\mathcal{C}_0$ which result in the coexistence of both colors.
\begin{align*}
&\mathrm{Var}(X)=\mathbb{E}[X^2]-\mathbb{E}[X]^2=\sum_{1 \leq i,j \leq n}\mathbb{E}[x_i \cdot x_j]-\mathbb{E}[X]^2=\\ &\sum_{i=1}^{n}\mathbb{E}[x_{i}^{2}] +\sum_{1\leq i\ne j \leq n}\mathbb{E}[x_i\cdot x_j]-\mathbb{E}[X]^2=
\mathbb{E}[X]+\sum_{1\leq i\ne j \leq n}\mathbb{P}[x_i=1\wedge x_j=1]-\mathbb{E}[X]^2=\\
&\mathbb{E}[X]+n(n-1)(1-p)^{2n-3}p_b^2-\mathbb{E}[X]^2=\mathbb{E}[X]+\mathbb{E}[X]^2((1-\frac{1}{n})\frac{1}{1-p}-1).
\end{align*}
Since $\mathbb{E}[X]=\omega(1)$, we have $\mathbb{E}[X]=o(\mathbb{E}[X]^2)$. Furthermore by using $p=o(1)$ we have 
\[
(1-\frac{1}{n})\frac{1}{1-p}-1=\frac{p}{1-p}-\frac{1}{n}\cdot \frac{1}{1-p}=\frac{pn-1}{n(1-p)}=o(1). 
\]
Putting both together we thus conclude that $\mathrm{Var}[X]=o(\mathbb{E}[X]^2)$.

Proof of part (ii): We want to show that in the majority model with $p_b=o(e^{np}/n)$ on $\mathcal{G}_{n,p}$ for $p\leq (1-\epsilon)\frac{\log n}{n}$ the process a.a.s. gets fully red. We first prove that there is no blue non-leaf node in configuration $\mathcal{C}_1$ a.a.s. (we handle the leaf nodes later). Let $v$ be an arbitrary node.
\begin{align*}
&\mathbb{P}[\mathcal{C}_1(v)=b\wedge d(v)\ne 1]=\mathbb{P}[\mathcal{C}_1(v)=b\wedge d(v)=0]+\sum_{d=2}^{n}\mathbb{P}[\mathcal{C}_1(v)=b\wedge d(v)=d]=\\
&\mathbb{P}[\mathcal{C}_1(v)=b|d(v)=0]\mathbb{P}[d(v)=0]+\sum_{d=2}^{n}\mathbb{P}[\mathcal{C}_1(v)=b|d(v)=d]\mathbb{P}[d(v)=d]\le\\
& p_b(1-p)^{n-1}+\sum_{d=2}^{n}{d+1 \choose \lfloor d/2\rfloor+1}p_b^{\lfloor d/2\rfloor+1}{n-1 \choose d}p^d(1-p)^{n-1-d}.
\end{align*}
In the last step, we used the fact that for a node of degree $d$ to get blue, at least $\lfloor d/2\rfloor+1$ nodes must be blue among the nodes in its neighborhood plus itself. Now, by applying $p_b=o(e^{np}/n)$, $1-x\le e^{-x}$, and ${d+1\choose \lfloor d/2\rfloor+1}\le 2^{2d}$, we have
\begin{align*}
&\mathbb{P}[\mathcal{C}_1(v)=b\wedge d(v)\ne 1]\le o(\frac{e^{np}}{n})e^{-np}e^p+\sum_{d=2}^{n}2^{2d}(\frac{e^{np}}{n})^{\lfloor d/2\rfloor+1}n^dp^de^{-np}e^{p(d+1)}\le\\
&o(\frac{1}{n})+\sum_{d=2}^{n}2^{6d}(\frac{e^{np}}{n})^{\lfloor d/2\rfloor}n^{d-1}p^d
\end{align*}
where we used $e^p\le e$ and $e^{(d+1)p}\le 2^{4d}$. Since the summation is maximized for $p=(1-\epsilon)\frac{\log n}{n}$ and $d/4\le \lfloor d/2\rfloor$ for $d\geq 2$, we have
\begin{align*}
&\mathbb{P}[\mathcal{C}_1(v)=b\wedge d(v)\ne 1]\le o(\frac{1}{n})+\sum_{d=2}^{n}\frac{1}{n^{\epsilon \lfloor d/2\rfloor}}\frac{(2^6\log n)^d}{n}\le o(\frac{1}{n})+\frac{1}{n}\sum_{d=2}^{n}(\frac{2^6\log n}{n^{\frac{\epsilon}{4}}})^d=\\
&o(\frac{1}{n})+\frac{1}{n}\sum_{d=2}^{\frac{8}{\epsilon}-1}(\frac{2^6\log n}{n^{\frac{\epsilon}{4}}})^d+\frac{1}{n}\sum_{d=\frac{8}{\epsilon}}^{n}(\frac{2^6\log n}{n^{\frac{\epsilon}{4}}})^d.
\end{align*}
The first summation is maximized for $d=2$ and the second one for $d=8/\epsilon$. Furthermore, $8/\epsilon$ is a constant. Therefore,
\begin{align*}
&\mathbb{P}[\mathcal{C}_1(v)=b\wedge d(v)\ne 1]\le o(\frac{1}{n})+ (\frac{8}{\epsilon}-2)\frac{2^{12}\log^2n}{n^{1+\frac{\epsilon}{2}}}+(n-\frac{8}{\epsilon}+1)\frac{2^{\frac{48}{\epsilon}}\log^{\frac{8}{\epsilon}}n}{n^3}=\\
&o(\frac{1}{n})+o(\frac{1}{n})+o(\frac{1}{n})=o(\frac{1}{n}).
\end{align*}
By the union bound, a.a.s. there is no blue non-leaf node in $\mathcal{C}_1$\footnote{We exclude the leaves since the expected number of blue leaves in $\mathcal{C}_1$ is equal to $n{n-1 \choose 1}p(1-p)^{n-2}p_b$, and for instance for $p=(1-\epsilon)\frac{\log n}{n}$ and $p_b=\frac{e^{np}}{n\log\log n}$ the expectation is equal to $\Omega(\log n/\log\log n)$  which is a growing function in $n$ (where we applied $e^{-x-x^2}\le 1-x$).}. Now building on this statement, we prove that $\mathcal{C}_2$ is fully red a.a.s. 

Let random variable $l(u)$ denote the number of leaves adjacent to node $u$. We know a.a.s. all non-leaf nodes are red in $\mathcal{C}_1$ (all the following probabilities are conditional on the occurrence of this event, but to lighten the notation we skip it). Thus, for an arbitrary node $u$ to be blue in $\mathcal{C}_2$, it must have at least one leaf neighbor and it also must be blue in $\mathcal{C}_0$ (otherwise, all its leaf neighbors will be red in $\mathcal{C}_1$). Let us consider two cases. First, $l(u)\geq 2$ and $\mathcal{C}_0(u)=b$. Second, $l(u)=1$, $\mathcal{C}_0(u)=b$, and $d(u)=1$, where we require $d(u)=1$ because otherwise $u$ has one or more non-leaf neighbors which are red in $\mathcal{C}_1$ and then $u$ is red in $\mathcal{C}_2$.
\begin{align*}
&\mathbb{P}[\mathcal{C}_2(u)=b]\le \mathbb{P}[\mathcal{C}_0(u)=b\wedge l(u)\ge 2]+\mathbb{P}[\mathcal{C}_0(u)=b\wedge l(u)=1 \wedge d(v)=1]\le\\
&{n-1 \choose 2}p^2(1-p)^{2n-5}p_b+{n-1 \choose 1}p(1-p)^{2n-4}p_b\le\\ &n^2p^2e^{-(2n-5)p}p_b+npe^{-(2n-4)p}p_b\le
o(\frac{1}{n})(\frac{n^2p^2}{e^{np}}+\frac{np}{e^{np}})=o(\frac{1}{n})
\end{align*}
where we used $p_b=o(\frac{e^{np}}{n})$ and $e^{5p},e^{4p}$ are upper-bounded by a constant, say $e^{5}$. Therefore by the union bound, there is no blue node in $\mathcal{C}_2$ a.a.s. which finishes the proof.
\end{proof}
\vspace{-0.35cm}
\begin{theorem}
\label{dynamo}
In $\mathcal{G}_{n,p}$ any dynamo is of size at least $(\frac{1}{2}-\frac{c}{\sqrt{np}})n$ for a large constant $c$ a.a.s.
\end{theorem}

\begin{proof}
The main idea of the proof is similar to the dense case in Theorem~\ref{theorem 1}. It suffices to prove that a.a.s. a set of size $s=(\frac{1}{2}-\delta)n$ for $\delta=\frac{c}{\sqrt{np}}$ cannot control a set of the same size. By definition of controlling, this implies no set of size $s$ or smaller can control a set of size $s$ or larger; consequently, there is no dynamo of size $s$ or smaller. We show that the probability that an arbitrary node set of size $s$ controls a set of the same size is so small that the union bound over all possibilities yields our claim. 

Let $S,S'$ be two node sets of size $s$. we want to bound the probability that $S$ controls $S'$. Since $\mathbb{E}[e(S',S)]=s^2p$, by applying the Chernoff bound (Theorem~\ref{Chernoff} (i)) and $\delta^2=\frac{c^2}{np}$ for a sufficiently large constant $c$, we have
\[
\mathbb{P}[(1+\delta)\mathbb{E}[e(S',S)]\le e(S',S)]\leq e^{-\frac{\delta^2\mathbb{E}[e(S,S')]}{3}}= e^{-\frac{c^2s^2p}{3np}}\le e^{-2n}.
\]
Similarly, since $\mathbb{E}[e(S',V\setminus S)]=s(n-s)p$,
\[
\mathbb{P}[e(S',V\setminus S)\le (1-\delta)\mathbb{E}[e(S',V\setminus S)]]\leq e^{-\frac{c^2s(n-s)p}{2np}}\le e^{-2n}.
\]
Furthermore,
\[(1+\delta)\mathbb{E}[{e(S',S)}]=(1+\delta)(\frac{1}{2}-\delta)^2n^2p\le (1-\delta)(\frac{1}{2}+\delta)(\frac{1}{2}-\delta)n^2p= (1-\delta)\mathbb{E}[e(S',V\setminus S)].
\] 
This implies $\mathbb{P}[e(S',S)\geq e(S',V\setminus S)]\leq 2e^{-2n}$. Furthermore, $\mathbb{P}[S\ \textrm{controls}\ S']\leq \mathbb{P}[e(S',S)\geq e(S',V\setminus S)]$ because if $e(S',S)< e(S',V\setminus S)$, then there is a node in $S'$ which shares more than half of its neighbors with $V\setminus S$. Therefore, $\mathbb{P}[S\ \textrm{controls}\ S']\le 2e^{-2n}$. By the union bound, the probability that there exist sets $S,S'$ of size $s$ such that $S$ controls $S'$ is upper-bounded by $2^{2n}2e^{-2n}=o(1)$, where $2^{2n}$ is an upper bound on the number of possibilities of choosing sets $S$ and $S'$.
\end{proof}
\section{Expanders}
\label{expanders}
Roughly speaking, our main goal in this section is to show that the majority model is an ``efficient'' and ``fast'' density classifier on regular expanders. Let us first state Lemma~\ref{mixing lemma}, which is our main tool. Recall that for a graph $G$ the second-largest absolute eigenvalue of its adjacency matrix is denoted by $\lambda(G)$ (to lighten the notation we simply write $\lambda$ where $G$ is clear from the context).
\begin{lemma}
\label{mixing lemma} (Lemma 2.3 in~\cite{hoory2006expander}) In a $\Delta$-regular graph $G=(V,E)$ for any two node sets $S,S'\subseteq V$, $|e(S,S')-\frac{|S||S'|\Delta}{n}|\leq \lambda\sqrt{|S||S'|}$.
\end{lemma}

In the above lemma, the left-hand side is roughly the deviation between the number of edges among $S$ and $S'$ in $G$ and the expected number of edges among $S$ and $S'$ in the random graph $\mathcal{G}_{n,\Delta/n}$ on the node set $V$. A small $\lambda$ (i.e., good expansion) implies that this deviation is small, so the graph is nearly random in this sense; in other words, the number of edges between any two node sets is almost completely determined by their cardinality. Intuitively, this implies in the majority model the number of blue nodes that a blue set can create in the next round is proportional to its size. We phrase this argument more formally in Lemma~\ref{lemma 1} and Lemma~\ref{lemma 2}. 
\begin{lemma}
\label{lemma 1}
In the majority model and $\Delta$-regular graph $G$, if $|B(t)|\leq (\frac{1}{2}-\frac{2\lambda}{\Delta})n$ then $|B(t+1)|\leq \frac{n}{4}$.
\end{lemma}
\begin{proof}
For each node in $B(t+1)$, the number of neighbors in $B(t)$ is at least as large as the number of neighbors in $R(t)$, which implies $e(B(t+1),R(t))\le e(B(t+1),B(t))$. Now, by applying Lemma~\ref{mixing lemma} to both sides of the inequality, we have
\[
\frac{|B(t+1)||R(t)|\Delta}{n}-\lambda\sqrt{|B(t+1)||R(t)|}\leq \frac{|B(t+1)||B(t)|\Delta}{n}+\lambda\sqrt{|B(t+1)||B(t)|}. 
\]
Dividing by $\sqrt{|B(t+1)|}$ and re-arranging the terms give
\[
\sqrt{|B(t+1)|}(|R(t)|-|B(t)|)\leq \frac{\lambda n}{\Delta}(\sqrt{|B(t)|}+\sqrt{|R(t)|}).
\] 
Since $|R(t)|-|B(t)|\geq \frac{4\lambda}{\Delta}n$ and $\sqrt{|B(t)|}+\sqrt{|R(t)|}\leq 2\sqrt{n}$, we have
\[
|B(t+1)|\frac{16\lambda^2n^2}{\Delta^2}\leq \frac{\lambda^2n^2}{\Delta^2}4n\Rightarrow |B(t+1)|\leq \frac{n}{4}.
\]
\end{proof}
\begin{lemma}
\label{lemma 2}
In the majority model and $\Delta$-regular graph $G$, $|B(t)|\leq \frac{n}{4}$ implies $|B(t+1)|\leq 16\frac{\lambda^2}{\Delta^2}|B(t)|$. 
\end{lemma}
\begin{proof}
Since each node in $B(t+1)$ must have at least $\Delta/2$ neighbors in $B(t)$, we have $\frac{|B(t+1)|\Delta}{2}\le e(B(t+1),B(t))$. Applying Lemma~\ref{mixing lemma} to the right side of the inequality gives
\[
\frac{|B(t+1)|\Delta}{2}\le \frac{|B(t+1)||B(t)|\Delta}{n}+\lambda\sqrt{|B(t+1)||B(t)|}\Rightarrow 
\]
\[
\sqrt{|B(t+1)|}(1-\frac{2|B(t)|}{n})\le \frac{2\lambda}{\Delta}\sqrt{|B(t)|}.
\]
Now, utilizing $\frac{|B(t)|}{n}\leq \frac{1}{4}$ and taking the square of both sides of the equation imply $|B(t+1)|\leq 16\frac{\lambda^2}{\Delta^2}|B(t)|$.
\end{proof}
Putting Lemma~\ref{lemma 1} and Lemma~\ref{lemma 2} together immediately provides Theorem~\ref{theorem 2}.
\begin{theorem}
\label{theorem 2}
In the majority model and $\Delta$-regular graph $G$, if $|B(0)|\leq (\frac{1}{2}-\frac{2\lambda}{\Delta})n$ then the process gets fully red in $\mathcal{O}(\log_{\Delta^2/\lambda^2}n)$ rounds.
\end{theorem}
\begin{corollary}
In the majority model and $\Delta$-regular graph $G$ with $\lambda(G)=o(\Delta)$, $|B(0)|\le (\frac{1}{2}-\delta)n$ for an arbitrary constant $\delta>0$ results in fully red configuration in sub-logarithmically many rounds.  
\end{corollary}
So far we proved our desired density classification property of the majority model on regular expanders. Now, we discuss that combining these results with some prior works yields some very interesting propositions, in particular solving an open problem by Peleg~\cite{peleg2014immunity}. 

The random $\Delta$-regular graph $\mathcal{G}_{n}^{\Delta}$ is the random graph with a uniform distribution
over all $\Delta$-regular graphs on $n$ vertices, say $[n]$. It is known~\cite{friedman2003proof} that a.a.s. $\lambda(\mathcal{G}_{n}^{\Delta})=\mathcal{O}(\sqrt{\Delta})$ for $\Delta\geq 3$. Therefore, Theorem~\ref{theorem 2} implies that in the majority model on $\mathcal{G}_{n}^{\Delta}$, if $|B(0)|\le (\frac{1}{2}-\frac{c}{\sqrt{\Delta}})n$ for some large constant $c$ then the process gets fully red a.a.s. This result is already known by G\"artner and Zehmakan~\cite{gartner2018majority}, however with a much more involved proof.

Recall that a graph is $(\alpha,\beta)$-immune if any node set of size $s\le \beta n$ controls at most $\alpha s$ nodes, and it is asymptotically optimally immune if it is $(\frac{c_2}{\Delta},\beta)$-immune for some constants $c_2,\beta>0$. As argued in the introduction, by~\cite{peleg2014immunity,gartner2018majority} we know that for any $\Delta>c_1$ for some constant $c_1$, there exists an asymptotically optimally immune $\Delta$-regular graph. However, it would be interesting to construct such graphs explicitly.
For $\Delta\ge \sqrt{n}$, Peleg~\cite{peleg2014immunity} established explicit
constructions by using structures for symmetric block designs, and he left the case of $\Delta<\sqrt{n}$ as an open problem. We settle this problem by exploiting a large family of regular Cayley graphs, called Ramanujan graphs. A $\Delta$-regular graph $G$ is \emph{Ramanujan} if $\lambda(G)=\sqrt{2\Delta-1}$. Ramanujan graphs are ``optimal'' expanders because Alon and Boppana~\cite{alon1986eigenvalues} proved that for a $\Delta$-regular graph $G$, $\lambda(G)\geq \sqrt{2\Delta-1}-o(1)$. Thus, Lemma~\ref{lemma 2} implies that for any $\Delta$-regular Ramanujan graph a node set of size $s\leq \frac{n}{4}$ can control at most $\frac{16\lambda^2}{\Delta^2}s= \frac{16(2\Delta-1)}{\Delta^2}s\le \frac{32}{\Delta}s$ nodes. This means that any $\Delta$-regular Ramanujan graph is $(\frac{1}{4},\frac{32}{\Delta})$-immune; i.e., it is asymptotically optimally immune.

\begin{theorem}
\label{theorem 3}
All regular Ramanujan graphs are asymptotically optimally immune.
\end{theorem}

Lubotzky, Phillips, and Sarnak~\cite{lubotzky1988ramanujan} showed that arbitrarily large $\Delta$-regular Ramanujan graphs exist when $\Delta-1$ is prime, and moreover they can be explicitly constructed (see also~\cite{marcus2015interlacing,morgenstern1994existence}). This result plus Theorem~\ref{theorem 3} answer the aforementioned question by Peleg.

Finally, as we argued regularity and expansion are sufficient properties for efficient density classification, but a natural question arises: are they also necessary? Some certain level of expansion seems to be needed for a graph to show such a density classification behavior under the majority model because otherwise there can exist a relatively small subset $S$ such that each node in $S$ has at least half of its neighbors in $S$; clearly, if $S$ is fully blue initially, it stays blue forever, even though all the remaining nodes are red. Regarding regularity, if the graph is not regular but almost regular, that is the minimum degree and the maximum degree differ by a constant factor, then the same proof ideas provide similar results. However, large degree gaps can lead into the state where a small subset of nodes of large degrees controls a large set of nodes of small degrees, which is in contrast with density classification. All in all, this would be an interesting question to be addressed rigorously in future work. 

\bibliography{ref}
\end{document}